\documentclass{lmcs}
\pdfoutput=1

\usepackage{lastpage}
\renewcommand{\dOi}{14(3:17)2018}
\lmcsheading{}{1--\pageref{LastPage}}{}{}%
{Jan.~17,~2018}{Sep.~05,~2018}{}

\usepackage[latin1]{inputenc} 
\usepackage[T1]{fontenc} 
\usepackage[square,numbers]{natbib} 
\usepackage{amsthm}
\usepackage{hyperref}  
\usepackage{prooftree}
\usepackage{url}
\usepackage{color}
\newcommand{\commento}[1]{}

\newcommand{\myit}[1]{\texttt{\textit{~#1}}}
\newcommand{\fjil}{\text{FJ\&}\lambda}
\newcommand{\A}{{\sf{A}}}
\newcommand{\B}{{\sf{B}}}
\newcommand{\C}{{\sf{C}}}
\newcommand{\D}{{\sf{D}}}
\newcommand{\Hi}{{\sf{J}}}
\newcommand{\I}{{\sf{I}}}
\newcommand{\Ei}{{\sf{E}}}
\newcommand{\T}{{\sf{T}}}
\newcommand{\U}{{\sf{U}}}
\newcommand{\f}{{\sf{f}}}
\newcommand{\g}{{\sf{g}}}
\newcommand{\m}{{\sf{m}}}
\newcommand{\te}{{\sf{t}}}
\newcommand{\x}{{\sf{x}}}
\newcommand{\y}{{\sf{y}}}
\newcommand{\CD}{{\sf{CD}}}
\newcommand{\ID}{{\sf{ID}}}
\newcommand{\K}{{\sf{K}}}
\newcommand{\M}{{\sf{M}}}
\newcommand{\N}{{\sf{N}}}
\newcommand{\Si}{{\sf{H}}}
\newcommand{\CT}{{\it{CT}}}
\newcommand{\va}{{\sf{v}}}
\newcommand{\vu}{{\sf{u}}}
\newcommand{\w}{{\sf{w}}}
\newcommand{\p}{{\sf{p}}}
\newcommand{\iInt}{\iota}
\newcommand{\Obj}{{\sf{Object}}}
\newcommand{\st}{<:}

\newcommand{\Int}{\tau}
\newcommand{\IntP}{\sigma}
\newcommand{\fInt}{\varphi}
\newcommand{\red}{\longrightarrow}
\newcommand{\substop}[1]{[#1]} 
\newcommand{\multisubst}[1]{\substop{#1}}
\newcommand{\tty}[2]{(#1)^{?#2}}
\newcommand{\ttyy}[2]{(#1)^{#2}}
\newcommand{\tInf}[2]{\texttt{tInf}(#1;#2)}
\newcommand{\der}[3]{#1\vdash#2:#3}
\newcommand{\tCk}[3]{\texttt{tCk}(#1;#2;#3)}
\newcommand{\derS}[3]{#1\vdash^*#2:#3}
\newcommand{\fields}[1]{\texttt{fields}(#1)}
\newcommand{\new}[2]{{\sf{new}}\,#1(\overrightarrow{#2})}
\newcommand{\newp}[2]{{\sf{new}}\,#1(#2)}
\newcommand{\met}[3]{#1.#2(\overrightarrow{#3})}
\newcommand{\cast}[2]{(#1)\,#2}
\newcommand{\clD}[3]{{\sf{class}}\,#1\,{\sf{extends}}\,#2\,{\sf{implements}}\,\overrightarrow#3\,\{\overline {\T}\, \overline {\f} ; \K\,\overline {\M}\}}
\newcommand{\inD}[3]{{\sf{interface}}\,#1\,{\sf{extends}}\,\overrightarrow#2\,\set{\overline#3;\!}}
\newcommand{\super}{{\sf{super}}}
\newcommand{\ret}{{\sf{return}}~}
\newcommand{\this}{{\sf{this}}}
\newcommand{\mtype}[2]{\texttt{mtype}(#1;#2)}
\newcommand{\mbody}[2]{\texttt{mbody}(#1;#2)}
\newcommand{\lambdaU}[2]{\overrightarrow{#1}\to{#2}}
\newcommand{\lambdaT}[3]{\overrightarrow {#2}\overrightarrow{#1}\to{#3}}
\newcommand{\signN}{\texttt{mh}}
\newcommand{\sign}[1]{\signN(#1)}
\newcommand{\sm}[4]{#1#2(\overrightarrow{#3}\overrightarrow {#4})}
\newcommand{\inDD}[4]{{\sf{interface}}\,#1\,{\sf{extends}}\,\overrightarrow#2\,\{\overline#3;  \overline#4\}}
\newcommand{\Asign}[1]{\texttt{A-mh}(#1)}
\newcommand{\Dsign}[1]{\texttt{D-mh}(#1)}
\newcommand{\Amtype}[2]{\texttt{A-mtype}(#1;#2)}
\newcommand{\Dmtype}[2]{\texttt{D-mtype}(#1;#2)}
\newcommand{\cond}[3]{#1?\,#2\!:\!#3}
\newcommand{\join}[2]{\texttt{lub}(#1,#2)}
\newcommand{\Bool}{\sf{boolean}}
\newcommand{\set}[1]{\{#1\}}
\newcommand{\OK}{{\sf{OK}}}
\newcommand{\true}{{\sf{true}}}
\newcommand{\false}{{\sf{false}}}
\newcommand{\dfn}{def.}
\newcommand{\E}{{\mathcal{E}}}
\newcommand{\IntR}{\rho}

\newcommand{\myrule}[2]{\text{\prooftree
#1
\justifies
#2
\endprooftree}}
\newcommand{\rn}[1]{\small{#1}}
\newcommand{\myruleN}[3]{\text{\prooftree
#1
\justifies
#2
\using \text{\rn{[#3]}}
\endprooftree}}

\begin{document}
\title[Java $\&$ Lambda: a Featherweight Story]{Java $\&$ Lambda: a Featherweight Story}
\author[]{Lorenzo Bettini\rsuper{a}}
\address{\lsuper{a}Dipartimento di Statistica, Informatica, Applicazioni, Universit\`a di
Firenze, Italy}
\email{lorenzo.bettini@unifi.it}
\address{\vskip-7pt}
\email{betti.venneri@unifi.it}

\author[]{Viviana Bono\rsuper{b}}
\address{\lsuper{b}Dipartimento di Informatica, Universit\`a di Torino, Italy}
\email{bono@di.unito.it}

\author[]{Mariangiola Dezani-Ciancaglini\rsuper{b}}
\address{\vskip-7pt}
\email{dezani@di.unito.it}

\author[]{Paola Giannini\rsuper{c}}
\address{\lsuper{c}Dipartimento di Scienze e Innovazione Tecnologica, Universit\`{a} del Piemonte Orientale,  Italy}
\email{paola.giannini@uniupo.it}

\author[]{Betti Venneri\rsuper{a}}

\thanks{Mariangiola Dezani was partially supported by EU H2020-644235 Rephrase project, EU H2020-644298 HyVar project, IC1402 ARVI and Ateneo/CSP project RunVar. Paola Giannini has the financial support of the Universit\`a  del Piemonte Orientale.}
\dedicatory{Dedicated to Furio Honsell on the occasion of his 60th birthday.}

\begin{abstract}
We present $\fjil$, a new core calculus that extends 
{\em Featherweight Java} (FJ) with interfaces,  supporting multiple inheritance in a restricted form,
{\em $\lambda$-expressions}, and {\em intersection types}. Our main goal is to formalise how
lambdas and intersection types are grafted on Java 8, by studying their
properties in a formal setting. We show how intersection types play a
significant role in several cases, in particular in the typecast of a
$\lambda$-expression and in the typing of conditional expressions.
We also embody interface \emph{default methods} in  $\fjil$,  since they increase the
dynamism of $\lambda$-expressions, by allowing these methods to be called on
$\lambda$-expressions.

The crucial point in Java 8 and in our calculus is that $\lambda$-expressions
can have various types according to the context requirements (\emph{target types}):
indeed, Java code does not compile when $\lambda$-expressions come without target types.
In particular, in the operational semantics we must record target
types by decorating $\lambda$-expressions, otherwise they would be lost in the runtime
expressions.

We prove the subject reduction property and progress for the resulting calculus, and 
we give a type inference algorithm that returns the type of a given program if
it is well typed.
The design of $\fjil$ has been driven by the aim of making it a subset of Java 8,
while preserving the elegance and compactness of FJ.
Indeed, $\fjil$ programs are typed and behave the same as Java programs.

\end{abstract}
\maketitle
 
\section{Introduction}
Currently Java is one of the most popular programming languages. Java offers
crucial features such as platform-independence and type-safety. Moreover it
continuously evolves with new features, while maintaining backward compatibility.
Following and sometimes influencing the evolution of Java development, programming language researchers have studied new features in the context of core calculi and formal models (see Section~\ref{rw}). This paper is a further step in this direction, focusing on \emph{intersection types} and \emph{Java 8}'s \emph{$\lambda$-expressions}. 
These two notions share a long and common history. In the recent past, intersection types have played a fundamental role in the construction and in the study of $\lambda$-calculus models, see Part III of~\cite{BDS13}. This successful marriage now acquires a new lease of life in Java 8.  $\lambda$-Expressions introduce a functional programming style on top of the object-oriented basis, while the typecast of intersection types gives  $\lambda$-expressions almost multiple identities. 

The background of the authors (which is similar to that of Furio Honsell) made them curious to understand how lambdas and intersection types have been grafted on a programming language like Java, by studying their properties in a formal calculus.
The obvious choice was to start from the \emph{Featherweight Java} (FJ) calculus, which has been proposed in~\cite{IPW01} as a minimal core language for modelling the essential aspects of Java's type system that are significant for the proof of type safety.  For our purposes, we extend FJ  by adding interfaces, with multiple inheritance, $\lambda$-expressions, and intersection types. 

Java 8 represents intersection by $\&$. Intersection types are introduced in a restricted form: they  can contain at most one class, which must be the first one specified, and multiple interfaces, provided that the intersection induces an unnamed class or interface. This means that the class and the interfaces cannot have methods with the same name and different types. We formalise the correctness requirements for building intersections through a function that gives the list of method headers defined in the type, with the condition that the same method name cannot get different signatures. We show how intersection types play a significant role in several cases, in particular in the typecast of a  $\lambda$-expression  and in the typing of conditional expressions.

$\lambda$-Expressions are {\em poly expressions} in Java 8. This means they can
have various types according to the context requirements. More specifically, the
contexts must prescribe \emph{target types} for $\lambda$-expressions: indeed,
Java code does not compile when $\lambda$-expressions come without target types.
Instead, standard expressions have unique types, which are determined entirely from their structure.  This combination of two different notions of typing requires bidirectional checking~\cite{PT00,DP00} and it has been the most critical issue in designing the type system of our calculus.
We point out that  Java avoids the introduction of $\lambda$-calculus function
types for $\lambda$-expressions, which would open the gates to structural
subtyping.  A target type can be either a functional interface (i.e., an
interface with a single abstract method)  or an intersection of interfaces that
induces a functional interface. According to this approach,  our definition of the subtype relation is  based on type names, with the addition of structural subtyping rules only on intersections (see Section~\ref{syntax}).

Concerning operational semantics,  we  must take into account that the reductions modify the contexts. Therefore, target types would be lost unless we record them. In order to have the subject reduction property, we decided to decorate $\lambda$-expressions by their target types: decorated terms appear at runtime only.

We also embody \emph{default methods} in interfaces, since they increase the dynamism of $\lambda$-expressions by allowing these methods to be called on $\lambda$-expressions. We 
discuss two aspects of conditional expressions. When both branches can be typed independently from the context, then Java uses intersection to build the type of the conditional as a least upper bound of branch types. Instead, in the presence of branches that are $\lambda$-expressions, these $\lambda$-expressions must have the target type that is prescribed by the context. 

Finally, we prove subject reduction and progress for the resulting calculus, dubbed $\fjil$ (Featherweight Java with intersection types and $\lambda$-expressions). 
We also give an inference algorithm that applied to a program, i.e., a class table and a term, returns (if any) the  type of the term. This algorithm takes into account the declarations in the table, which also induces the partial order between types. 

The design of $\fjil$ has been driven by the aim of making it a subset of Java
8, while preserving the elegance and compactness of FJ. Indeed, $\fjil$ programs are typed and behave the same as Java programs. Thus, our main result is to show how several significant novelties  are interwoven in Java 8 in a \emph{type-safe} way.

\paragraph{Outline.} We present $\fjil$ in three steps. The main part of the paper concentrates on  $\fjil$ without default methods in interfaces and conditional expressions. This part has a classical structure: syntax (Section~\ref{syntax}), lookup functions  (Section~\ref{lf}), operational semantics (Section~\ref{sem}), typing rules (Section~\ref{tr}) and properties (Section~\ref{aux}). The extensions to default methods in interfaces and conditional expressions are shown in Section~\ref{dm} and~\ref{cond}, respectively. Section~\ref{inf} details a type inference algorithm for the whole $\fjil$. Related works are discussed in Section~\ref{rw} and Section~\ref{fw} concludes with some hints to future research.

 \section{Syntax}\label{syntax}

We use $ \A, \B, \C, \D$ to denote  classes, 
$\I,\Hi$ to denote interfaces, 
$\T,  \U$ 
to denote nominal pre-types, i.e.,  either classes or interfaces; $\f,\g$ to denote field names; $\m$ to denote method names; $\te$ to denote terms; $\x,\y$ to denote variables, including the special variable $\this$.  We use $\overrightarrow{\I}$ as a shorthand for the list $\I_1,\ldots, \I_n$, $\overline{\M}$ as  a shorthand for the sequence $\M_1\ldots\M_n$, and similarly for the other names. The order in lists and sequences is sometimes unimportant, and this is clear from the context. In rules, we write both $\overline{\N}$ as a declaration and $\overrightarrow{\N}$  for some name $\N$: the meaning is that a sequence is declared and the list is obtained from the sequence adding commas. 
The notation $\overline{\T}\, \overline{\f};$  abbreviates $\T_1\f_1;\ldots\T_n\f_n;$ and $\overrightarrow{\T} \overrightarrow{\f}$  abbreviates $\T_1\f_1,\ldots,\T_n\f_n$ (likewise $\overrightarrow{\T} \overrightarrow{\x}$) and $\this.\overline{\f} = \overline{\f};$ abbreviates $\this.\f_1 = \f_1; \ldots \this.\f_n = \f_n;$. Sequences of interfaces, fields, parameters and methods are assumed to contain no duplicate names. The keyword $\super$, used only in constructor's body, refers to the superclass constructor.

\begin{figure}[b]
$\begin{array}{rclll}
\CD & ::= &\clD \C \D \I & \text{class declarations}\\\\
\ID & ::= & \inD \I \I \Si&\text{interface declarations}\\\\
\K  & ::= &\C(\overrightarrow{\T} \overrightarrow{\f})\{\super(\overrightarrow{\f}); \this.\overline{\f} = \overline{\f};\!\}  & \text{constructor declarations}
\\\\
\Si & ::= & \T  \m(\overrightarrow{\T}  \overrightarrow{\x}) & \text{header  declarations} \\\\
\M & ::= &  \Si\,\{\ret  \te;\!\} & \text{method declarations}
\end{array}$\caption{Declarations}\label{cd}
\end{figure}

Figure~\ref{cd} gives declarations: $\CD$ ranges over class declarations; $\ID$ ranges over interface declarations; $\K$ ranges over constructor declarations; $\Si$ ranges over method header (or abstract method) declarations; $\M$ ranges over method (or concrete method) declarations. This figure is obtained from Figure 19-1 of \cite{P02} by adding interfaces and method headers. A class declaration gives (in order) the class name, the superclass, the implemented interfaces, the typed fields, the constructor and the methods. An interface declaration gives the extended interfaces and the method headers. The arguments of the constructors correspond to the immutable values of the class fields. The inherited fields are initialised by the call to $\super$, while the new fields are initialised by assignments.  Headers relate method names with result and parameter pre-types. Methods are headers with bodies, i.e., return expressions. In writing examples we omit $\sf{implements}$ and $\sf{extends}$ when the lists of interfaces are empty and we use $\epsilon$ for the empty list.

$\Obj$ is a special class without fields and methods: it does not require a declaration.
A \emph{class table} $\CT$ is a mapping from nominal types to their declarations. A \emph{program} is a pair $(\CT,\te)$. In the following we assume a fixed class table.

\emph{Pre-types} (ranged over by $\Int,\IntP$) are either a nominal type or the intersection of:
\begin{itemize}
\item interfaces or
\item a class (in the leftmost position) with interfaces.
\end{itemize}   
Using  $\iInt$ to denote either an interface or an intersection of interfaces we define:
\[\Int::= \C \mid\iInt \mid\C\&\iInt\quad\text{ where }\quad\iInt::= \I \mid \iInt\&\I \]
The notation  $\C[\&\iInt]$ means either the class $\C$  
or the pre-type $\C\&\iInt$. 

\begin{figure}[tb]
	$\begin{array}{c}
	 \sign\Obj=\epsilon  \qquad\myrule{\CT(\I) = \inD \I \I \Si }{\sign  \I  = \overrightarrow{\Si} \uplus \sign{\overrightarrow{\I}} }\\\\
	 	\myrule{ 
		\myit{\CT}(\C) = \clD \C \D \I \quad
		\overline{\M}=\overline{\Si\, \{ \ret  \te;\! \}} 
		}{\sign{\C} = \overrightarrow{\Si} \uplus \sign \D\uplus \sign{\overrightarrow{\I}} } \\\\ 
\sign{\I_1,\ldots, \I_n} = \biguplus_{1\leq j\leq n} \sign{\I_j} \qquad		 \sign{\T_1\&\ldots \& \T_n} = \biguplus_{1\leq i\leq n} \sign{\T_i} 
	\end{array}$
	\caption{Function $\signN$}\label{sign}
\end{figure}

\begin{figure}[tb]
$\begin{array}{c}
\sf{class\,\C\,extends\,\Obj\,\set{\C(\,)\,\set{\super(\,);}\,\C\,m(\I\,\x)\set{\ret \x.n(\,);}}}\\
\sf{interface\,\I\,\set{ \C\, n(\,);}}   \qquad \qquad \qquad \sf{interface\,\Hi\,\set{ \C\, m(\,);}}\qquad \qquad \qquad \sf{interface\,\Ei\,\set{\,}}
\end{array}$
\caption{A Simple Class Table}\label{sct}
\end{figure}

To define types we use the partial function $\signN$ that maps pre-types to lists of method headers, considered as sets, see Figure~\ref{sign}. We need also to define $\signN$ for lists of interfaces. By $\biguplus$ we mean the union of lists of method headers that is defined only if no method name occurs in different headers. For instance taking $\C$, $\I$ and $\Hi$ as in Figure~\ref{sct} we get $\sf{\sign{\C\&\I}=\sign\C\biguplus\sign\I=\set{\C\,m(\I\,\x),\C\, n(\,)}}$, while $\sign{\C\&\Hi}=\sign\C\biguplus\sign\Hi$ is undefined, since both 
$\sign\C$ and $\sign\Hi$ contain method $\m$ with different argument lists.

\begin{defi}[Types]\label{types}
A pre-type $\Int$ is a type if $\sign\Int$ is defined.
\end{defi}

\noindent
In the following we will always restrict $\T,\U,\Int,\IntP$ to range over types. The typing rules for classes and interfaces (see Figure~\ref{icdtr}) assure that all nominal pre-types in a well-formed class table are types.

In the treatment of $\lambda$-expressions a special kind of types is handy. A \emph{functional type} is an interface or an intersection of interfaces which is mapped by  $\signN$ to a singleton, i.e.,  exactly to one method header. In other words, the type has only a single abstract method. We use $\fInt$ to range over functional types.

For example, with respect to the class table of Figure~\ref{sct} the pre-type $\C\&\I$ is a type, while $\C\&\Hi$ is not a type. Moreover, the type $\I\&\Ei$ is a functional type, while the type $\I\&\Hi$ is not a functional type.

\begin{figure}[tb]
$\begin{array}{ccc}
\begin{array}{llr}
\te ::= & & \text{terms} \\
& \va & \text{value}\\
& \x & \text{variable}\\
& \te.\f & \text{field access} \\
& \met \te \m \te & \text{met.  invoc.}\\
& \new  \C \te& \text{object}\\
& \cast\Int \te & \text{cast}
\end{array}
&\qquad&
\begin{array}{llr}
\va ::= & & \text{values}\\
& \w & \text{proper value}\\
& \lambdaU \p \te & \text{pure $\lambda$-expression}\\
\w ::= & & \text{proper values}\\
& \new \C \va & \text{object}\\
& \ttyy {\lambdaU \p \te} \fInt& \text{decorated $\lambda$-expression}\\
\p ::= & & \text{parameters}\\
& \x & \text{untyped}\\
&  \T\, \x & \text{typed}
\end{array}
\end{array}$\caption{Terms}\label{terms}
\end{figure}

Terms are defined in Figure~\ref{terms}: the differences with Figure 19-1 of \cite{P02} are the casting to intersections and the addition of $\lambda$-expressions. Inside the set of values (ranged over by $\va,\vu$) we distinguish \emph{proper values} (ranged over by $\w$): a pure $\lambda$-expression is a value, while a $\lambda$-expression decorated by a functional type is a proper value.  The functional type represents the \emph{target type} ~\cite{java8} (page 93) of the pure $\lambda$-expression: these proper values can occur only at run time. A parameter $\p$ of a $\lambda$-expression can be either untyped or typed, but the typing rules forbid to mix typed and untyped parameters in the same $\lambda$-expression. 
We use $\te_\lambda$ to range over pure $\lambda$-expressions. 


\begin{figure}[tb]
$\begin{array}{c}
\myruleN{\CT(\C) = \clD \C \D \I }{\C \st  \D \quad \C \st  \I_j  \quad \forall ~ \I_j \in \overrightarrow{\I}}{$\st\C$}\\\\
\myruleN{\CT(\I) = \inD \I \I  \Si }{\I \st  \I_j  \quad \forall ~ \I_j \in \overrightarrow{\I}}{$\st\I$}\qquad 
\T\st \Obj\,[\st\Obj]\\\\\myruleN{\Int \st  \T_i \quad \text{for all }1\leq i\leq n}{\Int \st  \T_1\&\ldots\&\T_n}{$\st\&$R} \qquad
\myruleN{\T_i \st  \Int \quad \text{for some }1\leq i\leq n}{\T_1\&\ldots\&\T_n\st \Int}{$\st\&$L} 
\end{array} $\caption{Subtyping}\label{subt}
\end{figure}

The \emph{subtype relation} $\st$ takes into account both the hierarchy between nominal types induced by the class table and the set theoretic properties of intersection. In fact $\st $ is the reflexive and transitive closure of the relation induced by the rules in Figure~\ref{subt}. Rule [$\st\&$R] formalises the statement in the last two lines of page 677 in~\cite{java8}.

Notice that  the requirement ``$\sign\Int$ defined'' (see Definition~\ref{types}) allows us to build a nominal class that is a subtype of $\Int$, as prescribed by the Java 8 Language Specification~\cite{java8} (pages 70-71). Dually the existence of a nominal class that is a subtype of $\Int$ assures $\sign\Int$ defined since rule [$\C$  \myit{OK}] in Figure~\ref{icdtr} requires $\sign\C$ defined and $\sign\Int\subseteq\sign\C$, see the proof of Lemma~\ref{ht}(\ref{ht2}). 

It is easy to notice that $\iInt\st\Obj\&\iInt\st\iInt$ for all $\iInt$, but we do not consider these types as equivalent, since $\iInt$ can be a functional type while $\Obj\&\iInt$ cannot. Moreover, 
in the presence of generic types, the type erasure of $\Obj\&\iInt$ differs from
the type erasure of $\iInt$. Our choice agrees with the aim of  designing $\fjil$ as a subset of Java.

\section{Lookup functions}\label{lf}

Following the definition of FJ (Figure 19-2 of \cite{P02}) the evaluation and typing rules of $\fjil$ use partial functions which give the set of fields of a class and the body of a method in a class. A difference is that the function which returns the type of a method takes as second argument a type instead of a class. This function takes advantage of 
the function  $\signN$, defined in Figure~\ref{sign}. Figure~\ref{aux} lists the lookup functions.

\begin{figure}[h]
	$\begin{array}{c}
	 \fields\Obj = \epsilon  \qquad
	 \myrule{ \begin{array}[b]{c}
		\CT(\C) = \clD \C \D \I \\
		\fields \D = \overrightarrow{\U}  \overrightarrow{\g}
		\end{array}
	}{\fields \C = \overrightarrow{\U}  \overrightarrow{\g},  \overrightarrow{\T}  \overrightarrow{\f}}  \\ \\	
	\myrule{
		\sm \T \m \T  \x  \in \sign\Int
		}{ \mtype \m \Int = \overrightarrow{\T} \rightarrow \T }  \qquad 
	\myrule{ \begin{array}[b]{c}
		\CT(\C) = \clD \C \D \I\\
		\T  \m( \overrightarrow{\U} \overrightarrow{\x}) \{ \ret  \te; \} \in \overrightarrow{\M}
		\end{array}}{\mbody \m \C = (\overrightarrow{\x},\te)} \\ \\
	\myrule{
		\begin{array}[b]{c}
		\CT(\C) = \clD \C \D \I\\
		\m\text{ is not defined in }\overrightarrow{\M}
		\end{array}}{\mbody \m \C = \mbody \m \D}	   
	\end{array}$
	\caption{Lookup Fields and Methods}\label{aux}	
\end{figure} 

 \section{Operational Semantics}\label{sem}

In typing the source code, Java uses for $\lambda$-expressions the types required by the contexts enclosing them.  These types are called {\em target types}.
This means that $\lambda$-expressions are {\em poly expressions}, i.e., they can
have different  types in different contexts, see page 93 of~\cite{java8}.
More precisely:

\begin{enumerate}
	\item\label{c1} the target type of a $\lambda$-expression that occurs as an
	actual parameter of a constructor call is the type of the field in the class declaration;
	\item\label{c2} the target type of a $\lambda$-expression that occurs as an
	actual parameter of a method call is the type of the parameter in the method declaration;
	\item\label{c3} the target type of a $\lambda$-expression that occurs as a return term of a method is the result type in the method declaration;
	\item\label{c4} the target type of a $\lambda$-expression that occurs as the body
of another $\lambda$-expression is the result type of the target type of the external $\lambda$-expression;
	\item\label{c5} the target type of a $\lambda$-expression that occurs as argument of a cast is the cast type.
		\end{enumerate}

\noindent
According to~\cite{java8} (page 602): ``It is a compile-time error if a lambda
expression occurs in a program in some place other than an assignment context, an invocation context (like (\ref{c1}), (\ref{c2}), (\ref{c3}) and (\ref{c4}) above), or a casting context (like (\ref{c5}) above).''
		
Clearly, by reducing field accesses and method calls with the rules of FJ (see Figure 19-3 of \cite{P02}) we lose the information on target types and we do not know how to type the $\lambda$-expressions in the resulting terms. For this reason, we modify these rules and we add the rules for method invocation on $\lambda$-expressions in such a way the $\lambda$-expressions are decorated by their target types in the evaluated terms. 
Technically, we use the mapping $\tty{\te}{\Int}$ defined as follows:
\[\tty{\te}{\Int}=\begin{cases}
    \ttyy \te \Int  & \text{if  $\te$ is a pure $\lambda$-expression}, \\
    \te  & \text{otherwise}
\end{cases}
\] 
The typing rules assure that if $\te$ is a pure $\lambda$-expression, then
$\Int$ is a functional type, i.e., reducing well-typed terms we only get decorated terms of the shape $\ttyy{\te_\lambda}\fInt$.

As usual  $\multisubst{\x\mapsto\te}$ denotes the substitution of $\x$ by $\te$ and it generalises to an arbitrary number of variables/terms as expected.

The notation $\overrightarrow \x\mapsto  {\tty {\overrightarrow \va} {\overrightarrow{\T}}}$ is short for $\x_1\mapsto\tty{\va_1}{\T_1},\ldots,\x_n\mapsto\tty{\va_n}{\T_n}$. 

\begin{figure}[tb]
\prooftree
\fields \C  = \overrightarrow \T~\overrightarrow \f
\justifies
\new \C \va.\f_j \red \tty{\va_j}{\T_j}
     %
           \using \text{[E-ProjNew]}\qquad
\endprooftree \prooftree
\C \st  \Int
\justifies
 \cast\Int{\new \C \va} \red \new \C \va
     %
           \using \text{[E-CastNew]}
\endprooftree

\prooftree
\mbody \m \C = (\overrightarrow \x, \te)\quad\mtype \m \C =\overrightarrow{\T} \rightarrow \T 
\justifies
\met{\new \C \va} \m  \vu \red
 \multisubst{\overrightarrow \x\mapsto  {\tty {\overrightarrow \vu} {\overrightarrow{\T}}}, \this\mapsto \new
  \C \va}\tty \te {\T}
     %
           \using \text{[E-InvkNew]}
\endprooftree

\prooftree
\mtype \m \fInt =\overrightarrow{\T} \rightarrow \T
\justifies
\met{\ttyy{\overrightarrow{\y}\to \te} \fInt}\m \va  \red
 \multisubst{\overrightarrow \y\mapsto \tty{\overrightarrow \va}{\overrightarrow \T}}\tty \te \T
     %
           \using \text{[E-Invk$\lambda$U]}
\endprooftree

\prooftree
\mtype \m \fInt=\overrightarrow{\T} \rightarrow \T
\justifies
\met{\ttyy{\overrightarrow{\T}\overrightarrow{\y}\to \te} \fInt}\m  \va \red
 \multisubst{\overrightarrow \y\mapsto \tty{\overrightarrow \va}{\overrightarrow \T}}\tty \te \T
     %
           \using \text{[E-Invk$\lambda$T]}
\endprooftree

$
 \cast{\fInt}{\te_\lambda} \red \ttyy{\te_\lambda}\fInt$
  [E-Cast$\lambda$] 
\qquad
\prooftree
\fInt \st  \fInt' 
\justifies
 \cast{\fInt'}{\ttyy{\te_\lambda} {\fInt}} \red \ttyy{\te_\lambda} {\fInt}
     %
           \using \text{[E-Cast$\lambda$Target]}
\endprooftree
\caption{Computational Rules}\label{comr}
\end{figure}

\begin{figure}[tb]
\prooftree
\te \red \te'
\justifies
\te.\f \red \te'.\f
           \using \text{[E-Field]}
\endprooftree \qquad
\prooftree
 \te \red \te'
\justifies
\met \te \m \te \red \met {\te'}\m \te
           \using \text{[E-Invk-Recv]}
\endprooftree \qquad
\prooftree
\te \red \te'
\justifies
\cast \Int \te \red \cast \Int{\te'}
           \using \text{[E-Cast]}
\endprooftree 

\prooftree
 \te \red \te'
\justifies
\w.\m(\overrightarrow \va, \te, \overrightarrow \te) 
   \red \w.\m(\overrightarrow \va, \te', \overrightarrow \te)
     %
           \using \text{[E-Invk-Arg]}
\endprooftree 

\prooftree
\te \red \te'
\justifies
\sf{new}\, \C(\overrightarrow \va, \te, \overrightarrow \te) 
   \red \sf{new}\, \C(\overrightarrow \va, \te', \overrightarrow \te)
     %
           \using \text{[E-New-Arg]}
\endprooftree 
\caption{Congruence Rules}\label{conr}
\end{figure}

The reduction rules are given in Figures~\ref{comr} and~\ref{conr}. It is easy
to verify that all pure $\lambda$-expressions being actual parameters or resulting terms in the l.h.s.\ are decorated by their target types in the r.h.s. Notably, in rules [E-Invk$\lambda$U] and [E-Invk$\lambda$T] we require the  pure $\lambda$-expression to come with its target type. This will be enforced by the typing rules. Coherently, the method receiver in rule [E-Invk-Arg] must be a proper value.

For example, with respect to the class table of Figure~\ref{sct}, we get:
\[\sf{new\,\C(\,).m(\epsilon\to\,new\,\C(\,))\red\ttyy{\epsilon\to\, new\,\C(\,)}\I.n(\,)\red\,new\,\C(\,)}\]

Since our interest is mainly in the type system, our reduction rules for $\lambda$-expressions strongly simplify the evaluation mechanism described in~\cite{java8} (Section 15.27.4). Our choice of decorating  $\lambda$-expressions with their target types during execution mimics the novelty of the Java 8 strategy for translating a $\lambda$-expression to bytecode.  
The Java 8 compiler, instead of generating an anonymous inner class, and then instantiating an object, replaces the creation of this extra class and object with a bytecode ``invokedynamic''  instruction. 
This instruction is used to delay the implementation of the $\lambda$-expression body until runtime, when the $\lambda$-expression is invoked for the first time.
In a similar way,  in the operational rules of our model, when a reduction step yields a $\lambda$-expression to be used,  
we have to record the target type which identifies the $\lambda$-expression in that context. So both approaches deal with $\lambda$-expressions only at run time when they are needed. 

 \section{Typing rules}\label{tr}

$\fjil$ adds intersection types and $\lambda$-expressions to FJ. This addition
requires non-trivial extensions of the typing rules, which are the main
contribution of the present paper. We start explaining the rules for terms shown
in Figure~\ref{sdtr}. Typing judgements are of the shape $\der \Gamma  \te \Int$
, where an environment $\Gamma$ is a finite mapping from variables to nominal types. We also use $\derS\Gamma  \te \Int$ as an abbreviation to simplify typing rules, as explained below.

\begin{figure}[b]
$\begin{array}{c}
 \myruleN{\x: \T \in \Gamma}{\der\Gamma  \x  \T} {T-VAR} 
  \qquad
 \myruleN{\der\Gamma  \te {\C[\&\iInt]} \quad \fields \C = \overrightarrow{\T}  \overrightarrow{\f} }{\der\Gamma { \te.\f_j } {\T_j}}{T-FIELD} \\
 \\
 \myruleN{
	\begin{array}[c]{c}
	\der \Gamma  \te \Int \quad
	\mtype \m \Int = \overrightarrow{\T} \rightarrow \T \quad
	\derS\Gamma  {\overrightarrow{\te}}{\overrightarrow{\T}}
	\end{array}
}{\der\Gamma {\met \te \m \te} \T}{T-INVK} \\
 \\
 \myruleN{
	\begin{array}[c]{c}
	\fields \C = \overrightarrow{\T}  \overrightarrow{\f} \quad
	\derS\Gamma {\overrightarrow{\te}}{ \overrightarrow{\T} }
	\end{array}
}{\der\Gamma{ \new  \C \te}\C}{T-NEW} \\
  \\
\myruleN{\sign\fInt=\T \m(\overrightarrow{\T} \overrightarrow{\x})\quad\derS{\Gamma,\overrightarrow{\y}:\overrightarrow{\T}} \te\T}{\der\Gamma {\ttyy{\overrightarrow{\y}\to \te} \fInt}\fInt}{T-$\lambda$U}\\
 \\
	\myruleN{\sign\fInt=\T \m(\overrightarrow{\T} \overrightarrow{\x})\quad\derS{\Gamma,\overrightarrow{\y}:\overrightarrow{\T}}\te \T	}{\der\Gamma{\ttyy{\overrightarrow{\T}\overrightarrow{\y}\to \te} \fInt}\fInt}{T-$\lambda$T}
\end{array}$
\caption{Syntax Directed Typing Rules}\label{sdtr}
\end{figure}

A first observation is that in field selections and method calls the receivers
can be typed by intersection types, and this is easily taken into account in
rules [T-FIELD] and [T-INVK].  Rules [T-$\lambda$U] and [T-$\lambda$T] deal with
$\lambda$-expressions, where types of parameters are omitted or explicitly
given, respectively. As already discussed,  a $\lambda$-expression is always
typed by the target type that is prescribed by the context. For this reason,
there are no rules for typing pure $\lambda$-expressions, instead rules
[T-$\lambda$U] and [T-$\lambda$T] type-check $\lambda$-expressions decorated by
their target types.

The main  technical innovation in our type system is the introduction of the judgement  $\vdash^*$.
It  is used in the typing rules of Figure~\ref{sdtr},
in order to play a specific role according
to whether the term being typed is a pure $\lambda$-expression or not. 

In the first case, the pure $\lambda$-expression can only have the target type that is mentioned by the context. Therefore, assigning a type to a  $\lambda$-expression by $\vdash^*$ means that the expression is decorated with this type and then it has the annotated type by the rules  [T-$\lambda$U] and [T-$\lambda$T]. These rules  require that the type is a functional type $\fInt$ and the term matches the signature of $\fInt$. Namely,  if  $\sign\fInt=\sm \T \m \T \x$, then, assuming the types $\overrightarrow \T$ for the parameters,  the body of the $\lambda$-expression  must have  type $\T$  according to  $\vdash^*$. 

Otherwise, when the term  is different from a pure $\lambda$-expression, the typing rules derive for the term the unique type that is induced by its syntactic structure. In this case, the  judgement  $\vdash^*$  has the role of checking whether  this type is a subtype of the one expected in the context. Therefore, judgements $\vdash^*$ turn out to be equivalent to standard subtype assertions, following the style of FJ where explicit subsumption is replaced by algorithmic subtype statements in the typing rules.

The above discussion suggests the following typing rules:
\[\myrule{ \der\Gamma {\te}\IntP \quad\IntP\st\Int }{\derS\Gamma {\te}\Int } \qquad\qquad\qquad\myrule{ \der\Gamma {\ttyy{\te_\lambda}\fInt}\fInt }{\derS\Gamma {\te_\lambda}\fInt } \]
which, taking advantage of the notation $\tty{\,}{}$, become:\label{star}
\[\myruleN{ \der\Gamma {\tty\te\Int}\IntP \quad\IntP\st\Int }{\derS\Gamma \te\Int }{$\vdash$ $\vdash^*$}   \]
Notice that if $\te$ is a decorated $\lambda$-expression only rules
[T-$\lambda$U] and [T-$\lambda$T]  can be applied, and then $\IntP=\Int$ and $\Int$ must be a functional type. In Java, $\lambda$-expressions can contain only final (or effectively final) variables from the enclosing environment, see page 607 of~\cite{java8}. We do not need to check this since $\fjil$ does not have assignments.\footnote{Note that in Java, $\lambda$-expressions can instead access any field of the containing class, whether final or not, and this extends naturally to $\fjil$.  Dealing with final or effectively final local variables in $\fjil$, after introducing assignments, would still be rather easy, since it would require to perform only a local control flow analysis in the scope of the containing method.}

The main feature of $\vdash^*$ is to simplify many of the typing rules in Figure~\ref{sdtr}. For instance without $\vdash^*$ we should write the rule for typing constructor calls with only one parameter as follows:\[
\myrule{
	\begin{array}[c]{c}
	\fields \C = \T\f \text{ if $\te$ is a pure $\lambda$-expression then 
	$\der\Gamma {\ttyy\te\T}\T$ else  $\der\Gamma \te\Int$ with $\Int\st\T$ }
	\end{array}
}{\der\Gamma{ \newp  \C \te}\C}{}\]
 We remark that we do not have two type systems, the judgement $\vdash^*$ is only a shorthand for an alternative between two possible judgements in the system $\vdash$.

Rules [T-INVK]  and [T-NEW] for $\fjil$ differ from the homonymous rules for FJ
in the use of the judgement $\vdash^*$ for actual parameters, and field initialisers, respectively. As usual $\derS\Gamma  {\overrightarrow{\te}}{\overrightarrow{\T}}$ is short for $\derS\Gamma  {\te_1}{\T_1}, \ldots, \derS\Gamma  {\te_n}{\T_n}$. Rules [T-$\lambda$U] and [T-$\lambda$T] are new but faithful to the statements in~\cite{java8} (see Section 15.27.3): as explained above, the body of the  $\lambda$-expression is typed by means of $\vdash^*$.

\begin{figure}[tb]
$\begin{array}{c}
	 \myruleN{\der\Gamma  \te  \Int \quad \Int \st  \IntP}{\der\Gamma {\cast\IntP \te } \IntP}{T-UCAST}\qquad
	\myruleN{\begin{array}{c}
	\der\Gamma {\ttyy{\te_\lambda}\fInt} \fInt 
	\end{array}}
	{\der\Gamma{\cast\fInt {\te_\lambda}} \fInt}{T-$\lambda$UCAST} \\
	\\
	 \myruleN{\der\Gamma  \te \Int \quad  \Int \sim \C[\&\iInt]\quad  \IntP\sim\D[\&\iInt']\quad \Int\not<:\IntP\quad\text{either }\C \st  \D\text{ or } \D \st  \C}{\der\Gamma {\cast\IntP t } \IntP} {T-UDCAST} 
\end{array}$
\caption{Cast Typing Rules}\label{ctr}
\end{figure}

We complete typing rules for terms by defining rules for type casts  in Figure~\ref{ctr}. The upcast rule [T-UCAST] is the natural generalisation of the homonymous rule of FJ to intersection types. Notice that the arguments of the casts in this rule and in rule [T-UDCAST] cannot be  pure $\lambda$-expressions, since a typing judgement must be derivable for them. The cast of a $\lambda$-expression $\te_\lambda$ requires to use the functional type as target type for  $\te_\lambda$, see rule [T-$\lambda$UCAST]. For example Java allows the cast $(\I\&\Ei)((\,)\to\newp\C{\,})$ but disallows the cast $(\Obj\&\I)((\,)\to\newp\C{\,})$ where $\I$, $\Ei$ and $\C$ are defined in Figure~\ref{sct}.

Rule [T-UDCAST] types casts which can fail, losing subject reduction. We use $\Int'\sim\IntP'$ as short for $\Int'<:\IntP'$ and $\IntP'<:\Int'$. The condition $\Int\not<:\IntP$ assures that this rule is not applied when  [T-UCAST] can be used. In rule [T-UDCAST] the cast of the class can be up or down, and the casts of the interfaces are possibly unrelated. Notice that this rule agrees with the prescriptions given in~\cite{java8} (Section 5.5.1) since there are no final classes in $\fjil$. As particular cases, all types can be sources of casts when the targets are intersections of interfaces and vice versa. In other words, if $\Int\sim\iInt$ for some $\iInt$, then $\IntP$ can be arbitrary and vice versa if $\IntP\sim\iInt'$ for some $\iInt'$, then $\Int$ can be arbitrary. In fact $\C$ or $\D$ can be  $\Obj$ and it is easy to verify that $\iInt\sim\Obj\&\iInt$ for all $\iInt$. The requirement ``the success of the cast is determined by the most restrictive component of the intersection type'' (see page 122 of~\cite{java8}) means that the classes in the intersections must be related by subtyping.

\begin{figure}[tb]
$\begin{array}{c}
 \myruleN{
	\begin{array}[b]{c}
	\derS{\overrightarrow{\x} : \overrightarrow{\T},  \this : \C } \te \T \quad
	\T  \m(\overrightarrow{\T}  \overrightarrow{\x})\in \sign \C
	\end{array}
}{\T  \m(\overrightarrow{\T}  \overrightarrow{\x}) \{ \ret\,  \te; \}  \myit{OK} \text{ in }  \C }{$\M$  \myit{OK}  in $\C$}\\
  \\ 
 \myruleN{
	\begin{array}[b]{c}
	\K = \C(\overrightarrow{\U}  \overrightarrow{\g},  \overrightarrow{\T}  \overrightarrow{\f}) \{ \super(\overrightarrow{\g});  \this.\overline{\f} = \overline{\f}; \} \quad
	\fields \D = \overrightarrow{\U}  \overrightarrow{\g} \quad \overrightarrow{\M}  \myit{OK} \text{ in }  \C \\
	\sign \C \quad
	\mtype \m \C  \text{ defined implies }\mbody \m \C \text{ defined}
	\end{array}
}{\clD \C  \D \I  \myit{OK} } {$\C$  \myit{OK}}\\\\
\myruleN{\sign \I }{\inD \I \I\Si  \myit{OK}}{$\I$  \myit{OK}}
\end{array}$
\caption{Method, Class  and Interface Declaration Typing Rules}\label{icdtr}
\end{figure}

We end this section by defining the rules for checking that method, class and
interface declarations are well formed (Figure~\ref{icdtr}).  For methods, the
only difference with respect to the corresponding rule of FJ is the use of $\vdash^*$ instead of $\vdash$ in rule [$\M$  \myit{OK}  in $\C$]. This allows us to type also methods whose return term is a $\lambda$-expression.  Furthermore, we observe that parameter types and return types of method declarations cannot be intersection types (according to Java specification). In the rules for classes and interfaces,  writing $\sign\T$ means that we require $\sign\T$ to be defined in the current class table. In this way we avoid to deal with additional requirements for the validity of method overriding (see Figure 19-2 of~\cite{P02}). The last condition in the premises of rule [$\C$  \myit{OK}] assures that a class implementing a set of interfaces contains the bodies of  all the abstract methods defined in those interfaces.

It is easy to verify that the class table of Figure~\ref{sct} is well formed.
An example of type derivation which uses this class table is:
\[\prooftree
\prooftree
\fields\C=\epsilon
\justifies
\der{}{\sf{new}\,\C(\,)}\C
\endprooftree
\ \ \mtype\m\C=\I\to\C
\prooftree
\prooftree
\sign\I= \C\,\sf{n}(\,)
\prooftree
\fields\C=\epsilon
\justifies
\der{}{\sf{new}\,\C(\,)}\C
\endprooftree
\justifies
\der{}{\ttyy{\epsilon\to\, \sf{new}\,\C(\,)}\I}\I
\endprooftree
\justifies
\derS{}{\epsilon\to\, \sf{new}\,\C(\,)}\I
\endprooftree
\justifies
\der{}{\sf{new\,\C(\,).m(\epsilon\to\,new\,\C(\,))}}\C
\endprooftree\]

\bigskip

To sum up, a program is well typed if the class table is well formed  and the
term has a type in the system $\vdash$ starting from the empty environment, using the declarations and the subtyping of the class table.

 \section{Subject Reduction and Progress}

The subject reduction proof of $\fjil$ essentially extends that of FJ~\cite{P02}
(Solution 19.5.1) taking into account intersection types and using the
flexibility of the $\vdash^*$ judgement. 
 The substitution lemma is shown simultaneously for both judgments $\vdash$ and $\vdash^*$. Instead subject reduction is proved only for $\vdash$ by induction on reductions. As usual, our type system enjoys {\em weakening}, i.e., $\der\Gamma\te\T$ implies $\der{\Gamma,\x:\U}\te\T$ and $\derS\Gamma\te\T$ implies $\derS{\Gamma,\x:\U}\te\T$.

\begin{lem}\label{ht}
\begin{enumerate}
\item \label{ht1}  If $\C[\&\iInt]\st\D[\&\iInt']$, then {\em $\fields\D\subseteq\fields\C$}.
\item\label{ht2}
If {\em $ \mtype \m \Int = \overrightarrow{\T} \rightarrow \T $}, then  {\em $\mtype \m \IntP = \overrightarrow{\T} \rightarrow \T$} for all $\IntP\st \Int$. 
\end{enumerate}
\end{lem}
\begin{proof}
(\ref{ht1}) Let $\iInt$ and $\iInt'$ be present and $\D$ be not $\Obj$, the proof in the other cases being simpler. From $\C\&\iInt\st\D\&\iInt'$ we get $\C\&\iInt\st\D$  by rule [$\st\&$R]. Then $\C\st\D$ by rule [$\st\&$L] (since $\iInt\st\D$ cannot hold).\\
(\ref{ht2}) By induction on the derivation of $\IntP\st \Int$ one can show that $\sign\Int\subseteq\sign\IntP$.
\end{proof}

\begin{lem}[Substitution for $\vdash^*$ and $\vdash$]\label{subst}
\begin{enumerate}
  \item \label{subst1}If $ \derS{\Gamma,\x:\T} \te \Int$ and $ \derS\Gamma  \va \T$, then \mbox{$\derS \Gamma {\multisubst{\x\mapsto 
  \tty{\va}{\T}} \te}\Int$.}
  \item \label{subst2}If $\der{\Gamma,\x:\T} \te \Int$ and $ \derS\Gamma  \va \T$, then $\der \Gamma {\multisubst{\x\mapsto \tty{\va}{\T}} \te}\IntP$ for some $\IntP\st \Int$. 
\end{enumerate}
\end{lem}
\begin{proof} (\ref{subst1}) and (\ref{subst2}) are proved by simultaneous 
induction on type derivations. 

(\ref{subst1}).  If  $\derS{\Gamma,\x:\T} \te \Int$,  then the last rule applied is [$\vdash$ $\vdash^*$].
We consider first the case of $\te$ being a pure $\lambda$-expression, and 
then $\te$ being any of the other terms.\\
{\bf Case $\te=\overrightarrow{\y}\to \te'$.}  
Then $\Int=\fInt$ and  the premise of rule [$\vdash$ $\vdash^*$] is $\der{\Gamma,\x:\T}{\ttyy{\overrightarrow{\y}\to \te'}\fInt}\fInt$.
By part (\ref{subst2}) of the induction hypothesis  we have that 
$\der{\Gamma}{\ttyy{\multisubst{\x\mapsto \tty{\va}{\T}} (\overrightarrow{\y}\to {\te'})}\fInt}\IntP$ 
for some $\IntP\st \fInt$. Since the last rule applied in the derivation must be  
[T-$\lambda$U], 
we get $\IntP=\fInt$. Using rule 
[$\vdash$ $\vdash^*$] we conclude
$\derS \Gamma {\multisubst{\x\mapsto \tty{\va}{\T} } (\overrightarrow{\y}\to {\te'})}\fInt$. The proof for the case $\te=\overrightarrow{\T}\overrightarrow{\y}\to \te'$ is similar.\\
{\bf Case $\te$ not a pure $\lambda$-expression.} The premise of rule [$\vdash$ $\vdash^*$] must be
$\der{\Gamma,\x:\T } \te\IntR$ for some $\IntR\st \Int$.  By part (\ref{subst2}) of the induction hypothesis
 we have that $\der \Gamma {\multisubst{\x\mapsto \tty{\va}{\T}} \te}\IntP$ for some $\IntP\st \IntR$. The
transitivity of $\st$ gives $\IntP\st \Int$. Applying rule [$\vdash$ $\vdash^*$] we conclude
$\derS \Gamma {\multisubst{\x\mapsto {\tty{\va}{\T}}} \te}\Int$.

(\ref{subst2}).  By cases on the last rule used in the derivation of $\der{\Gamma,\x:\T} \te \Int$.\\
{\bf Case [T-VAR].} $\der{\Gamma,\x:\T} \x \Int$ implies $\T=\Int$. The judgment $ \derS\Gamma  \va \T$ must be obtained by applying rule
[$\vdash$ $\vdash^*$] with premise $ \der{\Gamma}{  \tty{\va}{\T}} \IntP$ for some $\IntP\st \T$, as required.\\
{\bf Case [T-FIELD].} In this case  $\te= \te'.\f_j$ and\\
\centerline{$\myruleN{\der{\Gamma,\x:\T }{ \te' } \C\&\iInt \quad \fields \C = \overrightarrow{\T}  \overrightarrow{\f} }{\der{\Gamma,\x:\T }{ \te'.\f_j}{ \T_j}}{T-FIELD}$}\\
(the case in which $\&\iInt$ is missing is easier).\\
The induction hypothesis implies $\der \Gamma {\multisubst{\x\mapsto {\tty{\va}{\T}}} \te'}\IntR$ for some $\IntR\st \C\&\iInt$.
The subtyping rules of Figure \ref{subt} give $\IntR=\D[\&\iInt']$ for some $\D\st\C$ and $\iInt'$. By Lemma \ref{ht}({\ref{ht1}}) we have that
$\fields\C\subseteq\fields\D$ and then $\T_j \f_j\in  \fields \D$. Therefore applying rule [T-FIELD] we conclude 
${\der{\Gamma}{ \multisubst{\x\mapsto {\tty{\va}{\T}}}\te'.\f_j}{ \T_j}}$.\\
{\bf Case [T-INVK].} In this case  $\te= {\met {\te'}\m \te }$ and\\
\centerline{$ \myruleN{
	\begin{array}[c]{c}
	\der{\Gamma,\x:\T}{ \te' } {\Int'} \quad
	\mtype \m {\Int'} = \overrightarrow{\T} \rightarrow \T' \quad
	\derS{\Gamma,\x:\T } {\overrightarrow{\te}}{\overrightarrow{\T} }
	\end{array}
}{\der{\Gamma,\x:\T}  {\met {\te'}\m \te } \T'}{T-INVK}
 $}
\\
Part (\ref{subst1}) of the induction hypothesis on 
$\derS{\Gamma,\x:\T } {\overrightarrow{\te}}{\overrightarrow{\T} }$ implies
\mbox{$\derS\Gamma {\multisubst{\x\mapsto {\tty{\va}{\T}}} \overrightarrow{\te} }{ \overrightarrow{\T}}$.}
By induction hypothesis on $\der{\Gamma,\x:\T}{ \te' } {\Int'}$ we have that
\mbox{$\der \Gamma {\multisubst{\x\mapsto {\tty{\va}{\T}}} \te'}\IntR$} for some $\IntR\st \Int'$.
Lemma \ref{ht}({\ref{ht2}}) gives
$\mtype \m {\IntR} = \overrightarrow{\T} \rightarrow \T' $. Applying rule [T-INVK] we conclude 
$\der{\Gamma}  {\multisubst{\x\mapsto {\tty{\va}{\T}}} (\met {\te'}\m \te) } \T'$.
\\
{\bf Case [T-NEW].} By part (\ref{subst1}) of the induction hypothesis  on the judgments for the parameters.
\\
{\bf Case [T-$\lambda$U].} In this case $\Int=\fInt$ and   $\te=\ttyy{\overrightarrow{\y}\to \te'}\fInt$ and\\
\centerline{$ \myrule{
	\sign\Int=\sm {\T'}\m \T \x\quad\derS{\Gamma,\x:\T,\overrightarrow{\y}:\overrightarrow{\T}} {\te'}{\T'}
}{\der{\Gamma,\x:\T}  {\ttyy{\overrightarrow{\y}\to \te'}\fInt} \fInt}
 $}
\\
By part (\ref{subst1}) of the induction hypothesis $\derS{\Gamma,\overrightarrow{\y}:\overrightarrow{\T} } {\multisubst{\x\mapsto {\tty{\va}{\T}}}{\te'}}{\T'} $.
Applying rule [T-$\lambda$U]
we conclude $\der{\Gamma}{\ttyy{\multisubst{\x\mapsto {\tty{\va}{\T}}} (\overrightarrow{\y}\to {\te'})}\fInt}\fInt$.\\
The proof for the rule [T-$\lambda$T] is similar.
\end{proof}

\begin{lem}\label{tb}
If {\em $ \mtype \m \C = \overrightarrow{\T} \rightarrow \T $} and {\em $ \mbody \m \C = (\overrightarrow{\x}, \te )$}, then
  \begin{align*}
    \derS{\overrightarrow{\x}:\overrightarrow{\T}, \this:\D} \te\T
  \end{align*}
for some $\D$ such that $\C\st\D$. 
\end{lem}
\begin{proof}
By definition of $\texttt{mbody}$, the method $\m$ must be declared either in
class $\C$ or in some class $\D$ which is a superclass of $\C$. In both cases rule [$\M$  \myit{OK}  in  $\C$] of Figure~\ref{icdtr} gives the desired typing judgement.
\end{proof}
\begin{lem}\label{key}
If $\derS\Gamma \te\Int$, then $\der\Gamma{ \tty \te \Int}\IntP$ for some $\IntP\st \Int$. 
\end{lem}
\begin{proof}The judgment $ \derS\Gamma  \te \Int$ must be obtained by applying rule
[$\vdash$ $\vdash^*$] with premise $ \der{\Gamma}{  \tty{\te}{\Int}} \IntP$ for some $\IntP\st \Int$, as required.\end{proof}

\begin{thm}[Subject Reduction]\label{subred}
If $\der\Gamma \te\Int$ without using rule \mbox{\em [T-UDCAST]} and $\te\red \te'$, then $\der\Gamma{ \te'}\IntP$  for some $\IntP\st \Int$.
\end{thm}
\begin{proof}
By induction on a derivation of $\te\red \te'$, with a case analysis on the final rule. 
We only consider interesting cases.

  \medskip

{\bf Case} \prooftree
\fields \C = \overrightarrow \T~\overrightarrow \f
\justifies
\new \C \va.\f_j \red \tty{\va_j}{\T_j}
               \using \text{[E-ProjNew]}
\endprooftree\\[3pt]
The l.h.s.\ is typed as follows:\\
\centerline{\prooftree
\prooftree
\fields \C  = \overrightarrow \T~\overrightarrow \f \quad \derS\Gamma {\overrightarrow \va }{ \overrightarrow \T}
\justifies
\der \Gamma {\new \C \va} \C
\endprooftree
\justifies
\der\Gamma {\new \C  \va. \f_j}{ \T_j}
\endprooftree}\\[3pt]
By Lemma~\ref{key} $\derS\Gamma { \va_j }{  \T_j}$ implies $\der\Gamma {\tty{
\va_j}{\T_j} } {\IntP}$ for some $\IntP\st\T_j$. 

  \medskip

{\bf Case} \prooftree
\mbody \m \C  = (\overrightarrow \x, \te'')\quad\mtype \m \C =\overrightarrow{\T} \rightarrow \T 
\justifies
\met{\new \C\va }\m  \vu   \red
 \multisubst{\overrightarrow \x\mapsto  \tty {\overrightarrow \vu} {\overrightarrow{\T}}, \this\mapsto \new\C\va }\tty {\te''} {\T}
           \using \text{[E-InvkNew]}
\endprooftree\\[3pt]
The l.h.s.\ is typed as follows: \\
\centerline{
\prooftree
\der\Gamma {\new \C  \va} \C\quad\mtype \m \C =\overrightarrow{\T} \rightarrow \T \quad\derS\Gamma { \overrightarrow \vu }{ \overrightarrow \T}
\justifies
\der\Gamma {\met{\new \C  \va} \m  \vu}\T
\endprooftree}
\\[3pt]
By Lemma~\ref{tb} $\mbody \m \C= (\overrightarrow \x, \te'')$ implies
$\derS{\overrightarrow \x:\overrightarrow \T, \this:\D}{\te''}\T$ with
$\C\st\D$. From $\der\Gamma {\new \C  \va} \C$ and $\C\st\D$ we get $\derS\Gamma {\new \C  \va} \D$.\\  By Lemma~\ref{key} $\der{\overrightarrow \x:\overrightarrow \T, \this:\D}{\tty {\te''} \T}\IntP$ for some $\IntP\st\T$. By Lemma~\ref{subst}(\ref{subst2}) and weakening\\\centerline{$\der\Gamma{\multisubst{\overrightarrow \x\mapsto  {\tty {\overrightarrow \vu} {\overrightarrow{\T}}}, \this\mapsto \new \C\va}\tty{\te''} {\T}}\IntR$ for some $\IntR\st\IntP$}
Finally by transitivity of $\st$ we have $\IntR\st\T$.
  
  \medskip

{\bf Case} \prooftree
\mtype \m \fInt =\overrightarrow{\T} \rightarrow \T
\justifies
\met{\ttyy{\overrightarrow{\y}\to {\te''} } \fInt}\m  \va   \red
 \multisubst{\overrightarrow \y\mapsto \tty{\overrightarrow \va}{\overrightarrow \T}}\tty {\te''}  \T
           \using \text{[E-Invk$\lambda$U]}
\endprooftree\\[3pt]
The l.h.s. is typed as follows:
\[ \prooftree
\prooftree
\sign\fInt=\T \m(\overrightarrow{\T} \overrightarrow{\x})\quad\derS{\Gamma,\overrightarrow \y: \overrightarrow \T } {\te''}  \T
\justifies
\der\Gamma {\ttyy{\overrightarrow{\y}\to {\te''} } \fInt}\fInt
\endprooftree
\mtype \m \fInt =\overrightarrow{\T} \rightarrow \T \quad\derS\Gamma {\overrightarrow \va }{ \overrightarrow \T}
\justifies
\der\Gamma{\met{ \ttyy{\overrightarrow{\y}\to {\te''} } \fInt}\m \va}\T
\endprooftree\]
By Lemma~\ref{key} $\der{\Gamma,\overrightarrow \y:\overrightarrow \T}{\tty {\te''}  \T}\IntP$ for some $\IntP\st\T$. By Lemma~\ref{subst}(\ref{subst2}) we derive $\der\Gamma{\multisubst{\overrightarrow \y\mapsto \tty{\overrightarrow \va}{\overrightarrow \T}}\tty{\te''}  \T}\IntR$ for some $\IntR\st\IntP$.
Finally by transitivity of $\st$ we have $\IntR\st\T$.

\medskip

{\bf Case} \prooftree
 \te \red \te'
\justifies
\w.\m(\overrightarrow \va, \te, \overrightarrow \te) 
   \red \w.\m(\overrightarrow \va, \te', \overrightarrow \te)
     %
           \using \text{[E-Invk-Arg]}
\endprooftree\\ 
The l.h.s.\ is typed as follows:\\
\centerline{\prooftree
\der\Gamma {\w} \Int\quad\mtype \m \Int =\overrightarrow{\T} \rightarrow \T \quad
\derS\Gamma { \overrightarrow \va}{ \overrightarrow {\T_{\va}}}\quad
\derS\Gamma {\te}{\T'}\quad
\derS\Gamma {\overrightarrow \te}{ \overrightarrow {\T_{\te}}}
\justifies
\der\Gamma {\w.\m(\overrightarrow \va, \te, \overrightarrow \te)}\T
\endprooftree}
where $\overrightarrow{\T}= \overrightarrow {\T_{\va}},\T',\overrightarrow{\T_{\te}}$. By Lemma~\ref{key} $\derS\Gamma {\te}{\T'}$ implies $\der\Gamma {\tty\te{\T'}}{\IntP}$ for some $\IntP\st\T'$. Since $ \te \red \te'$ implies that $\te$ cannot be a $\lambda$-expression we get $\tty\te{\T'}=\te$. By induction hypothesis $\der\Gamma {\te'}{\IntR}$ for some $\IntR\st\IntP$. Being $\IntR\st\T'$ applying rule [$\vdash$ $\vdash^*$]
 we derive $\derS\Gamma {\te'}{\T'}$. Therefore using the typing rule
\rn{[T-INVK]} we conclude $\der\Gamma {\w.\m(\overrightarrow \va, \te, \overrightarrow \te)}\T$.
\end{proof}

Rule [T-UDCAST] breaks subject reduction already for FJ, as shown in \cite{P02} (Section 19.4). Following \cite{P02} we can recover subject reduction by erasing the condition ``either $\C \st  \D$ or $\D \st  \C$'' in rule [T-UDCAST]. In this way the rule becomes:
\[\myruleN{\der\Gamma  \te \Int \quad \Int\not<:\IntP }{\der\Gamma {\cast\IntP \te } \IntP} {T-STUPIDCAST}\]

The closed terms that are typed without using rule  [T-UDCAST] enjoy the standard progress property. This can be easily proven by just looking at the shapes of well-typed irreducible terms.

\begin{thm}[Progress]
If $\der{} \te\Int$ without using rule {\em [T-UDCAST]} and $\te$ cannot reduce, then $\te$ is a proper value.
\end{thm}

Using rule  [T-UDCAST] we can type casts of proper values which cannot be reduced, like, for example, $\cast\C{(\sf{new}\, \Obj (\;))}$ with $\C$ different from $\Obj$. An example involving a $\lambda$-expression is $\cast\C{\ttyy{\epsilon \to \sf{new}\,\Obj(\,)}{\I}}$, where $\I$ is the interface with the only signature $\Obj\,\m(\;\;)$. This run-time term can be obtained by reducing $\cast\C{\cast{\I}{(\epsilon \to \sf{new}\,\Obj(\,))}}$.

To characterise the stuck terms (i.e., the irreducible terms which can be obtained by reducing typed terms and are not values) we resort to the notion of evaluation context, as done in \cite{P02} (Theorem 19.5.4). {\em Evaluation contexts} $\E$ are defined as expected: 
\[\E::= [\;] \mid \E.\f\mid \met \E\m\te\mid \w.\m (\overrightarrow \va, \E, \overrightarrow\te)\mid \sf{new}\, \C(\overrightarrow \va, \E, \overrightarrow\te)\mid(\Int)\E\]
Stuck terms are evaluation contexts with holes filled by casts of typed proper values which cannot reduce, i.e., terms of the shapes $\cast \Int {\new\C\va}$ with $\C\not\st\Int$ and $\cast\Int{\ttyy{\te_\lambda}\fInt}$ with $\fInt\not\st\Int$. Notice that $\cast {\A[\&\iInt]} {\new\C\va}$ cannot be typed when $\A,\C$ are unrelated classes. Instead rule [T-UDCAST] allows us to type all terms of the shape $\cast\Int{\ttyy{\te_\lambda}\fInt}$, when $\ttyy{\te_\lambda}\fInt$ has a type.

 \section{Default Methods}\label{dm}

This section is devoted to the extension of interfaces with default methods.
This extension shows the expressivity of casting $\lambda$-expressions to
functional types, whose definition is also changed, see below. For simplicity,
we omit the keyword $\sf{default}$ assuming that all methods implemented in
interface declarations are default methods, while any method terminated by a
semicolon is an abstract method.
This is a slight difference with respect to the syntax of Java, where the
$\sf{default}$ key is mandatory if an interface method has a body, and an
interface method lacking the $\sf{default}$ modifier is implicitly abstract.
Note that in Java, providing a body without the $\sf{default}$ modifier leads to
a compilation error.\footnote{Indeed, the presence of a method with body but
without the $\sf{default}$ modifier in an interface with a single abstract method also makes
the Java compiler bailout: the interface is not considered as a functional
interface at all.}

The first obvious modification is interface declaration, which includes also
method bodies:\[\ID ::= \inDD \I \I \Si \M\] This new interface declaration
requires to distinguish between methods defined in interfaces with or without
implementations. For this reason we consider two mappings from pre-types to
method headers, called $\texttt{A-mh}$ and $\texttt{D-mh}$, see
Figure~\ref{asds}.

\begin{figure}[tb]
$\begin{array}{c}
\prooftree
\CT(\I)=\inDD \I \I \Si \M
\justifies
\Asign \I = \overrightarrow{\Si}\uplus \Asign{\overrightarrow{\I}}
\endprooftree\\\\
\prooftree
\CT(\I)=\inDD \I \I \Si \M\quad \overline{\M}=\overline{\Si'\, \{ \ret \te; \}}
\justifies
\Dsign \I= \overrightarrow{\Si'}\uplus \Dsign{\overrightarrow{\I}}
\endprooftree\\ \\
\Asign{\I_1,\ldots ,\I_n} =  \Asign{\I_1\&\ldots \& \I_n} = 
	 \Asign{\C\&\I_1\&\ldots \& \I_n} = \biguplus_{1\leq i\leq n} \Asign{\I_i}\\\\
	 \Dsign{\I_1,\ldots ,\I_n} = \Dsign{\I_1\&\ldots \& \I_n} =\Dsign{\C\&\I_1\&\ldots \& \I_n} = \biguplus_{1\leq i\leq n} \Dsign{\I_i} \\
\text{if $\Dsign{\I_j}\cap\Dsign{\I_\ell}\not=\epsilon$ implies either $\I_j\st\I_\ell$ or $\I_\ell\st\I_j$}
	\end{array}$ \caption{Functions $\texttt{A-mh}$ and $\texttt{D-mh}$}\label{asds}
	\end{figure}
The mapping $\texttt{A-mh}$ gives the headers of abstract methods (without implementations) and the mapping $\texttt{D-mh}$ gives the headers of default methods (with implementations). For an interface $\I$ the set $\Asign\I$ contains the method headers defined in the declaration of $\I$ and those inherited,  the set $\Dsign\I$ contains the headers of the methods implemented in the declaration of $\I$ and those inherited. We also need to define  $\texttt{A-mh}$ and  $\texttt{D-mh}$ for lists of interfaces and for intersection pre-types. Java allows multiple inheritance from interfaces and intersections of interfaces only when there is no ambiguity in the definition of implemented methods. This is reflected in the conditions for the definition of $\texttt{D-mh}$. 
 
 The definition of $\signN$ for classes remains the same, although classes can inherit method bodies from interfaces. For a list of interfaces  (which can be a single interface) $\sign{\overrightarrow\I}$ is the union of $\Asign{\overrightarrow\I}$ and $\Dsign{\overrightarrow\I}$, when they do not contain the same method name, see page 292 of~\cite{java8}. 
 \[\begin{array}{ll}
	\sign {\overrightarrow\I} = \Asign {\overrightarrow\I} \uplus\Dsign {\overrightarrow\I}&\text{ if } \Asign{\overrightarrow\I} \cap \Dsign{\overrightarrow\I}=\emptyset
	\end{array}\]
For an intersection pre-type we take the  union $\uplus$ of the method headers in the class (if any) and those in the list of interfaces:
 \[\begin{array}{ll}
 \sign{\I_1\&\ldots\&\I_n}=\sign{\I_1,\ldots,\I_n}\qquad\qquad
 \sign{\C\&\I_1\&\ldots\&\I_n}=\sign\C\uplus\sign{\I_1,\ldots,\I_n}
	\end{array}\]	
In the above definitions  we use $\uplus$ to avoid the same method name with different signatures, as we explained  in discussing the function $\signN$ in Section~\ref{syntax}.	

The definition of types is unchanged, while a type is a \emph{functional type}
if it is an interface or an intersection of interfaces and it is mapped by
$\texttt{A-mh}$ to a singleton, see~\cite{java8} page 321.
Therefore, an interface (intersection of interfaces) having a single abstract
method can have several default methods. We observe this, for example, by
looking at the Oracle documentation of the \texttt{Function} functional
interface.\footnote{This functional interface has a single abstract method,
\texttt{apply}, and two default methods, \texttt{compose} and \texttt{andThen}
(\url{https://docs.oracle.com/javase/8/docs/api/java/util/function/Function.html}).}
We still use $\fInt$ to
range over functional types.

 The change of method headers naturally reflects on the lookup functions for method types. We now need two functions, $\texttt{A-mtype}$ and  $\texttt{D-mtype}$ for types:
 \[\myrule{\sm \T  \m \T \x\in
		  \Asign \Int 
		}{ \Amtype \m  \Int  = \overrightarrow{\T} \rightarrow \T } 
	\qquad	\qquad 
	\myrule{
		\sm \T  \m \T \x \in \Dsign \Int
		}{ \Dmtype \m  \Int = \overrightarrow{\T} \rightarrow \T }\]
 while the definition of $\sf{mtype}$ remains the same, but it uses the new function $\signN$. 
 
When looking for method bodies, we need to take into account also default methods defined in interface declarations, see Section 9.4.1 and  pages 522, 532 of~\cite{java8}. Figure~\ref{mbl} gives the new definition of the function $\texttt{mbody}$. In the rule for lists and intersections of interfaces we choose the implementation given in the smallest interface (which must be unique). In the rule for intersections we first consider the implementation given in the class (if any) and then those in the interfaces. 
 \begin{figure}[tb]
 $\begin{array}{c}
 $\myrule{ \begin{array}[b]{c}
		\CT(\C) = \clD \C \D \I \\
		\T  \m( \overrightarrow{\U} \overrightarrow{\x}) \{ \ret  \te; \} \in \overrightarrow{\M}
		\end{array}}{\mbody\m  \C = (\overrightarrow{\x},\te)} $\\ \\
	$\myrule{
		\begin{array}[b]{c}
		\CT(\C) = \clD \C \D \I \\
		\m\text{ is  not  defined in } \overrightarrow{\M} \quad  \mbody \m  \D \text{  is defined}
		\end{array}}{\mbody \m  \C  = \mbody \m  \D  }$\\ \\
		$\myrule{
		\begin{array}[b]{c}
		\CT(\C) = \clD \C \D \I \\
		\m\text{ is  not  defined in } \overrightarrow{\M} \quad  \mbody \m \D \text{ is  not  defined}
		\end{array}}{\mbody \m  \C  = \mbody \m{\overrightarrow{\I} }}$\\ \\
 \prooftree
\CT(\I) =\inDD \I \I \Si \M\quad  \sm \T  \m \T \x \{ \ret  \te; \} \in \overrightarrow{\M}\justifies
\mbody \m  \I  = (\overrightarrow{\x},\te)
\endprooftree\\\\ 
		\prooftree
\CT(\I) =\inDD \I \I \Si \M\quad  \m\text{ is  not  defined  in } \overrightarrow{\M}
\justifies
\mbody \m  \I = \mbody \m{\overrightarrow{\I} }
\endprooftree\\\\
		 \mbody \m{\I_1,\ldots ,\I_n} = \mbody \m {\I_1\&\ldots \& \I_n}  =\mbody \m {\I_j}\\
		\text{ if  
		$\mbody \m{\I_\ell}$ defined  
		implies }\I_j\st\I_\ell
		\\ \\
		 \mbody \m {\C\&\I_1\&\ldots \& \I_n}  = \begin{cases}
 \mbody \m \C     & \text{if defined }\\
   \mbody \m{\I_1,\ldots ,\I_n}    & \text{otherwise}
\end{cases}
		\end{array}$\caption{Method Body Lookup}\label{mbl}
		\end{figure}

 \bigskip

		The reduction of a method call on a $\lambda$-expression distinguishes the case of abstract methods from that of default methods, see Figure~\ref{ncr}. These rules replace rules [E-Invk$\lambda$U] and [E-Invk$\lambda$T] of Figure~\ref{comr}.
		\begin{figure}[tb]
$\begin{array}{c}		\prooftree
\Amtype \m \fInt =\overrightarrow{\T} \rightarrow \T\
\justifies
\met{\ttyy{\overrightarrow{\y}\to \te} \fInt}\m  \va  \red
 \multisubst{\vec \y\mapsto \tty{\vec \va}{\vec \T}}\tty \te \T
                \using \text{[E-Invk$\lambda$U-A]}
\endprooftree\\ \\
\prooftree
\Amtype \m \fInt =\overrightarrow{\T} \rightarrow \T
\justifies
\met{\ttyy{\overrightarrow{\T}\overrightarrow{\y}\to \te} \fInt}\m  \va   \red
 \multisubst{\vec \y\mapsto \tty{\vec \va}{\vec \T}}\tty \te \T
           \using \text{[E-Invk$\lambda$T-A]}
\endprooftree\\ \\
\prooftree
\mbody \m \fInt  = (\overrightarrow \x, \te)\quad\Dmtype \m \fInt =\overrightarrow{\T} \rightarrow \T
\justifies
\met{\ttyy{\te_\lambda} \fInt}\m \va  \red
 \multisubst{\vec \x\mapsto  \tty {\vec \va} {\overrightarrow{\T}}, \this \mapsto \ttyy{\te_\lambda} \fInt}\tty \te {\T}
           \using \text{[E-Invk$\lambda$-D]}
\endprooftree
\end{array}$\caption{New Computational  Rules}\label{ncr}
 \end{figure}
 
 The typing rules for $\lambda$-expressions must use $\texttt{A-mh}$ instead of $\texttt{sign}$:
 \[\begin{array}{c}\myruleN{\Asign\fInt=\T \m(\overrightarrow{\T} \overrightarrow{\x})\quad\derS{\Gamma,\overrightarrow{\y}:\overrightarrow{\T}} \te\T}{\der\Gamma {\ttyy{\overrightarrow{\y}\to \te} \fInt}\fInt}{T-$\lambda$UD}\\
 \\
	\myruleN{\Asign\fInt=\T \m(\overrightarrow{\T} \overrightarrow{\x})\quad\derS{\Gamma,\overrightarrow{\y}:\overrightarrow{\T}}\te \T	}{\der\Gamma{\ttyy{\overrightarrow{\T}\overrightarrow{\y}\to \te} \fInt}\fInt}{T-$\lambda$TD}\end{array}\]
 
The well-formedness condition of interfaces is as expected:				 
\[
 \myruleN{\overline{\M}  \myit{OK} \text{ in }  \I\quad\sign \I
 }{\inDD \I \I \Si \M \myit{OK}}{$\I$  \myit{OK}}\]
where
\[\myruleN{
	\begin{array}[b]{c}
	\derS{\overrightarrow{\x} : \overrightarrow{\T},\this:\I } \te \T \quad
	\T  \m(\overrightarrow{\T}  \overrightarrow{\x})\in \Dsign \I
	\end{array}
}{\T  \m(\overrightarrow{\T}  \overrightarrow{\x}) \{ \ret\,  \te; \}  \myit{OK} \text{ in }  \I }{$\M$  \myit{OK} \text{ in } \I}\]
Notice that $\this$ is typed by an interface, see page 480 of~\cite{java8}. 

\bigskip

For example, if we modify the class table of Figure~\ref{sct} by defining\[\sf{interface \,\Hi\,\set{\Obj\,\m(\,)\,\set{\ret\,new\,\Obj\,(\,);}}}\] we can type $(\I\&\Hi)(\epsilon\to\sf{new}\,\C\,(\,))$ by $\I\&\Hi$. To call the method $\m$ defined in $\Hi$ we can use $(\I\&\Hi)(\epsilon\to\sf{new}\,\C\,(\,))$ as receiver. Notice that we cannot use $(\Hi)(\epsilon\to\sf{new}\,\C\,(\,))$ as receiver since this term has no type. In fact $\Hi$ is not a functional type. Notice that to run this example in Java one need to add the keyword $\sf{default}$ in front of the declaration of method $\m$.

\bigskip

The subject reduction proof smoothly extends by replacing Lemma~\ref{tb} by the following lemma, which takes into account default methods in interfaces.

\begin{lem}\label{tbd}
  If {\em $ \mtype \m \Int = \overrightarrow{\T} \rightarrow \T $} and {\em $ \mbody \m \Int = (\overrightarrow{\x}, \te )$}, then
  \begin{align*}
  \derS{\overrightarrow{\x}:\overrightarrow{\T}, \this:\U} \te\T
  \end{align*}
  for some $\U$ such that $\Int\st\U$. 
\end{lem}

The evaluation contexts and the stuck terms are unchanged. 

\section{Conditional}\label{cond}

To consider conditional expressions we add the primitive type $\Bool$ to the set of types,  the boolean literals $\true$, $\false$ to the set of proper values and $\cond \te {\te_1} {\te_2}$ to the set of terms. Notice that $\Bool$ cannot be argument of an intersection, since the function $\signN$ for $\Bool$ is undefined\footnote{The addition of $\Bool$ instead of $\sf Boolean$ is simpler in many respects. We avoid to consider the fields and methods of $\sf Boolean$. By definition $\sf Boolean$ could occur in intersections, while $\Bool$ cannot. Being $\sf Boolean$ a final class we cannot instantiate $\C,\D$ in rule [T-UDCAST] by $\sf Boolean$, since for example $(\I)\,\true$ does not compile for any interface $\I$.}.

The reduction rules for conditionals are as expected:
\[\cond \true {\te_1} {\te_2}\red\te_1~\text{[E-IfTrue]}\qquad\cond \false {\te_1} {\te_2}\red\te_2~\text{[E-IfFalse]}\qquad
\prooftree{\te\red\te'}\justifies{\cond \te {\te_1} {\te_2}\red\cond {\te'} {\te_1} {\te_2}}\using\text{[E-If]}\endprooftree
\]

Intersection types are especially meaningful for the typing of conditional expressions, see page 587 of~\cite{java8}\footnote{In previous versions of Java the two branches were required to have types related by $<:$, see page 531 of~\cite{P02}.}. The key observation is that the term $\cond \te {\te_1} {\te_2}$ can reduce to either $\te_1$ or $\te_2$, therefore we can assure on the resulting term only what  $\te_1$ and $\te_2$ share. Let $\C$ be the minimal common superclass and $\I_1,\ldots,\I_n$ be the minimal common super-interfaces of $\Int_1$, $\Int_2$. Formally we require: \begin{itemize}
\item $\Int_1<:\C$ and $\Int_2<:\C$;
\item $\Int_1<:\D$ and $\Int_2<:\D$ imply $\C<:\D$;
\item $\Int_1<:\I_i$ and $\Int_2<:\I_i$ for $1\leq i\leq n$;
\item $\I_i\not<:\I_j$ and $\I_j\not<:\I_i$ for $1\leq i\not=j\leq n$ ;
\item $\Int_1<:\Hi$ and $\Int_2<:\Hi$ imply $\I_i<:\Hi$ for some $i \,(1\leq i\leq n)$.
\end{itemize}
If $\te$ has type $\Bool$, and $\te_1$ and $\te_2$ have types $\Int_1$, $\Int_2$, respectively, then Java derives type $\C\&\I_1\&\ldots\&\I_n$ for $\cond \te {\te_1} {\te_2}$. By defining $\join{\Int_1}{\Int_2}=\C\&\I_1\&\ldots\&\I_n$, this observation leads us to formulate the following typing rule for conditional expressions:
\[\myruleN{
	\der\Gamma  \te \Bool \quad \der\Gamma {\te_1}{\Int_1}  \quad \der\Gamma {\te_2}{\Int_2}
		}{\der\Gamma {\cond \te {\te_1} {\te_2}}  {\join{\Int_1}{\Int_2} }}{T-COND}\]
		For example, if we extend the class table of Figure~\ref{sct} with the
		declarations \[\sf{class\,\A\,extends\,\C\,\set{\cdots}}\ \ \sf{class\,\B\,extends\,\A\,implements\,\I\,\set{\cdots}}\ \ \sf{class\,\D\,extends\,\C\,implements\,\I\,\set{\cdots}}\]
		we get $\join\B\D=\C\&\I$.
		
We can easily check that $\sign{\join{\Int_1}{\Int_2}}$ is always defined. In fact by construction $\Int_i\st\sign{\join{\Int_1}{\Int_2}}$ and this implies $\sign{\join{\Int_1}{\Int_2}}\subseteq\sign{\Int_i}$ for $i=1,2$, see the proof of Lemma~\ref{ht}(\ref{ht2}).

We notice that our definition of $\sf{lub}$ is much simpler that the one in~\cite{java8} (pages 73-74-75), since $\fjil$ does not have generic types.
		
Rule [T-COND] clearly does not apply when one of the two branches of the
conditional is a $\lambda$-expression, or when one of the two branches of the conditional is in turn a conditional with a branch which is a $\lambda$-expression, and so on. In these cases, Java types the conditional only if it is has a target type. According to~\cite{java8} (page 587): ``A reference conditional expression is a poly expression if it appears in an assignment context or an invocation context.'' We do not strictly follow this requirement: we consider target types of conditionals only for the $\lambda$-expressions which appear in their branches. Clearly this does not modify the typability of terms. To render this typing we extend the mapping $\tty{\,}{\Int}$ to  conditional expressions by applying it to conditional branches:\\
\centerline{$\tty {\cond \te {\te_1} {\te_2}}\Int =
    \cond \te {\tty{\te_1}\Int }{\tty{\te_2}\Int}
$}

This assures that rule [T-COND] can be applied to conditional expressions having $\lambda$-expressions as branches.	In this way we formalise the sentence ``\ldots a conditional expression appears in a context of a particular kind with target type $\fInt$, its second and third operand expressions similarly appear in a context of the same kind with target type $\fInt$'', see page 587 of~\cite{java8}.

For example, using the class table of Figure~\ref{sct} and the above declaration of class $\B$ under the assumption that class $\B$ has no field, we can derive:

{\footnotesize{\[\prooftree
\prooftree
\fields\C=\epsilon
\justifies
\der{}{\sf{new}\,\C(\,)}\C
\endprooftree
\ \ \mtype\m\C=\I\to\C
\prooftree
\prooftree
\der{}\true\Bool
\prooftree
\sign\I= \C\,\sf{n}(\,)
\prooftree
\fields\C=\epsilon
\justifies
\der{}{\sf{new}\,\C(\,)}\C
\endprooftree
\justifies
\der{}{\ttyy{\epsilon\to\, \sf{new}\,\C(\,)}\I}\I
\endprooftree
\prooftree
\fields\B=\epsilon
\justifies
\der{}{\sf{new}\,\B(\,)}\B
\endprooftree
\justifies
\der{}{\sf{\cond\true{\ttyy{\epsilon\to\,new\,\C(\,)}\I}{new\,\B(\,)}}}\I
\endprooftree
\justifies
\derS{}{\sf{\cond\true{\epsilon\to\,new\,\C(\,)}{new\,\B(\,)}}}\I
\endprooftree
\justifies
\der{}{\sf{new\,\C(\,).m(\cond\true{\epsilon\to\,new\,\C(\,)}{new\,\B(\,)})}}\C
\endprooftree\]}}

\noindent
being $\join\I\B=\I$.

\bigskip
	
	The proof of subject reduction can be easily extended, since we can show:	
	\begin{lem}\label{keyc}
If $\derS\Gamma {\cond \te {\te_1} {\te_2}}\Int$, then $\derS\Gamma {\tty {\cond \te {\te_1} {\te_2}} \Int}\Int$. 
\end{lem}

In characterising stuck terms we need to add the evaluation context for conditionals: $\cond\E{\te_1}{\te_2}$.  
\section{Type Inference}\label{inf}
Our type inference algorithm naturally uses the technique of {\em bidirectional checking}~\cite{PT00,DP00}. In fact the judgement $\vdash$ operates in synthesis mode, propagating typing upward from subexpressions, while the judgement $\vdash^*$ operates in checking mode, propagating typing downward from enclosing expressions.

We assume a given class table to compute the lookup functions and the subtyping relation.
The partial function $\tInf\Gamma \te$ gives (if any) the 
type $\Int$ such that $\der\Gamma \te \Int$. It is always undefined for pure $\lambda$-expressions. It uses the predicate $\tCk\Gamma \te \Int$ which is true if $\derS\Gamma \te \Int$, i.e., according to rule [$\vdash$ $\vdash^*$] (see page~\pageref{star}):
\[\tCk\Gamma \te \Int \text{ if }\tInf\Gamma {\tty \te\Int}=\IntP \text{ and } \IntP\st\Int\]
Figure~\ref{inff} defines $\texttt{tInf}$: it just uses the rules of Figure~\ref{sdtr} without the rules for $\lambda$-expressions,  the rules of Figure~\ref{ctr}, the typing rules for $\lambda$-expressions of Section~\ref{dm} and the typing rules for the conditionals of Section~\ref{cond}.

We use $\tCk \Gamma { \overrightarrow \te}{ \overrightarrow \T}$ as short for $\tCk \Gamma { \te_1}{ \T_1},\ldots,\tCk \Gamma { \te_n}{ \T_n}$.

\begin{figure}[tb]
$\begin{array}{llll}
\tInf\Gamma \x & = & \T &\text{if } \x:\T\in\Gamma\\
\tInf\Gamma \true & = & \Bool \\
\tInf\Gamma \false & = & \Bool \\
\tInf\Gamma {\te.\f} & = & \T &\text{if } \tInf\Gamma \te=\C[\&\iInt]\text{ and } \T\f\in\fields \C\\
\tInf\Gamma {\new \C \te} & = & \C &\text{if }\fields \C=\overrightarrow{\T}\overrightarrow{\f} \text{ and } \tCk \Gamma { \overrightarrow \te}{ \overrightarrow \T}\\
\tInf\Gamma {\met  \te \m \te} & = & \T &\text{if }\tInf\Gamma \te=\Int\text{ and } \mtype  \m \Int=\overrightarrow{\T}\to \T \\&&&\text{and } \tCk \Gamma { \overrightarrow \te}{ \overrightarrow \T}\\
    \tInf\Gamma{\cast\Int \te} & = & \Int & \text{if  one of the following conditions holds}\\
&&&\bullet\quad\tCk\Gamma \te \Int\\
 &&&\bullet\quad\Int=\C[\&\iInt] \text{ and } \tInf\Gamma \te=\D[\&\iInt']\\
 &&&\phantom{\bullet\quad} \text{and either $\C\st \D$ or $\D\st \C$}\\
     \tInf\Gamma{\ttyy{\lambdaU \y \te}\fInt} & = &\fInt& \text{if }\Asign\fInt=\sm \T \m \T \x\text{ and }\tCk{\Gamma,\overrightarrow{\y}:\overrightarrow{\T}} \te \T\\
        \tInf\Gamma{\ttyy{\lambdaT \y \T \te}\fInt} & = &\fInt& \text{if }\Asign\fInt=\sm \T \m \T \x\text{ and }\tCk{\Gamma,\overrightarrow{\y}:\overrightarrow{\T}} \te \T\\
         \tInf\Gamma{\cond \te {\te_1} {\te_2}} & = & \Int & \text{if }\tCk\Gamma  \te  \Bool\text{ and }        \tInf \Gamma  {\te_1 } ={\Int_1} \\
         &&&  \text{and } \tInf\Gamma { \te_2}= {\Int_2}\text{ and }\Int=\join{\Int_1}{ \Int_2}\\
 \end{array}$
 \caption{Type Inference Function}\label{inff}
 \end{figure}

Building on Figure~\ref{icdtr} and the well-formedness rules for interfaces and
their methods of Section~\ref{dm}, Figure~\ref{wct} defines a predicate $\OK$ which  tests well-formedness of class tables, i.e., of classes, interfaces and methods. We use the following abbreviations: \dfn\  for defined, $\OK(\overrightarrow{\M}, \T)$ for $\OK(\M_1, \T),\ldots,\OK(\M_n, \T)$, and $\OK(\overrightarrow{\M})$ for $\OK(\M_1),\ldots,\OK(\M_n)$.

\begin{figure}[tb]
$\begin{array}{ll}
\OK(\T  \m(\overrightarrow{\T}  \overrightarrow{\x}) \{ \ret\,  \te; \} ,\U) &\text{if }\T  \m(\overrightarrow{\T}  \overrightarrow{\x})\in \sign \U\text { and}\\ &\tCk{\Gamma,\overrightarrow{\x}:\overrightarrow{\T},  \this : \U} \te \T\\
\OK(\clD \C  \D \I )&\text{if }\K = \C(\overrightarrow{\U}  \overrightarrow{\g},  \overrightarrow{\T}  \overrightarrow{\f}) \{ \super(\overrightarrow{\g});  \this.\overline{\f} = \overline{\f}; \} \\
	&\text{and }\fields \D = \overrightarrow{\U}  \overrightarrow{\g} \text{ and } \OK(\overrightarrow{\M}, \C) \text{ and}\\&
	\sign \C \text{ \dfn\ and for any }\m\\&
	\mtype \m \C  \text{ \dfn\ implies }\mbody \m \C \text{ \dfn}\\
\OK(\inDD \I \I \Si \M)&\text{if }	\OK(\overrightarrow{\M}, \I) \text{ and }\sign \I \text{ \dfn}\\
\OK(\overline\C\,\overline\I)&\text{if }	\OK(\overrightarrow{\C}) \text{ and }\OK(\overrightarrow{\I})
\end{array}$
\caption{Well-formedness Function}\label{wct}
\end{figure} 

\bigskip

For example, if we use the class table of Figure~\ref{sct} and we apply the inference function to the empty environment and to the term $\sf{new\,\C(\,).m(\epsilon\to\,new\,\C(\,))}$ we get $\tInf{}{\sf{new\,\C(\,)}}=\C$ and $\mtype\m\C=\I\to\C$. This requires $\tCk{}{\sf{\epsilon\to\,new\,\C(\,)}}\I$, which means\\ \centerline{$\tInf{}{\ttyy{\sf{\epsilon\to\,new\,\C(\,)}}\I}=\Int$ for some $\Int\st\I$.} Being $\Asign\I= \C\,\sf{n}(\,)$ and $\tInf{}{\sf{new\,\C(\,)}}=\C$ we derive $\tInf{}{\ttyy{\sf{\epsilon\to\,new\,\C(\,)}}\I}=\I$. We can then conclude $\tInf{}{\sf{new\,\C(\,).m(\epsilon\to\,new\,\C(\,))}}=\C$. Clearly this computation corresponds to the derivation shown at the end of Section~\ref{tr}.

 \section{Related Work}\label{rw}

The literature on object-oriented programming, in particular the literature on Java, is enormous. We only mention here some papers that, 
like the present one, introduce core calculi in order to enlighten relevant aspects of the object-oriented paradigm.

The seminal paper of Fisher, Honsell, and Mitchell~\cite{FHM94} presents one of
the first typed calculi modelling a fully-fledged \emph{object-based} language, and distilling ten years of studies on objects-as-records. Abadi and Cardelli in their encyclopaedic book~\cite{AC96} discuss foundational calculi of objects: the untyped calculus and the calculi with first-order, second-order and higher-order types. Extensions of the calculi in~\cite{AC96} have been used to formalise object behaviours. For example, object ownership and nesting between objects are the subject of~\cite{CNP01}. Castagna~\cite{Castagna97} provides a foundation for object-oriented languages focusing on overloading and multiple dispatch.

A formal description of the operational semantics and type system of a substantial subset of {\em Java} is the content of~\cite{DEK99}. The approach taken by Igarashi, Pierce and Wadler~\cite{IPW01} is instead to omit many features of Java obtaining an elegant and  small calculus, i.e., FJ,
suitable for extensions and variations. Today we can safely claim that 
this goal has been fully achieved. By using FJ,  generic classes are formalised in~\cite{IPW01}, a true module system is constructed in~\cite{AZ01}, inner classes are modelled in~\cite{IP02}, the existence of principal typings is shown in~\cite{AZ04}, transactional mechanisms are discussed in~\cite{JVWH05}, union types are proposed in~\cite{Igarashi07a}, cyclic objects with coinductive operations are introduced in~\cite{AZ12}, a co-contextual type checker is described in~\cite{KEBBM17}. The authors themselves have widely used FJ to formalise  extensions of Java with additional features, aiming at dynamic flexibility~\cite{WrapJavaJot,BCV08,BBV2010}, at assuring safety of communications~\cite{DDMY09,BCDGV13} and at enhancing code reuse under several aspects~\cite{BETTINI2013907, BETTINI2017419}.
%
FJ has been shown to be suitable also in dealing with semantics. We mention the denotational semantics in a theory of types and names~\cite{S01}, the type-preserving compilation into an intermediate language~\cite{LST02}, the coinductive big-step operational semantics~\cite{A12}, the semantics based on intersection types and approximants~\cite{RB14}. 

The benefits of {\em intersection types} to model multiple inheritance in
class-based languages were already shown by Compagnoni and Pierce
in~\cite{CP96}. B\"uchi and Weck~\cite{BW98} introduce the notion of compound types as anonymous reference types, expressed as a list of a class and various interfaces, so that objects having these types can combine  the behavioural specifications of several nominal types. They illustrate a rather interesting scenario which motivates the need of extending Java 1 with compound types. Two alternative ways for emulating compound types on the Java virtual machine are discussed. Furthermore,  the soundness of the proposal is verified with the theorem prover Isabelle/HOL.
In~\cite{BL08} an intersection type assignment system provides a program logic for the first order calculus of~\cite{AC96}. Intersection types are also  employed to synthesise mixins, which permit reuse of object-oriented code avoiding the ambiguities of
multiple inheritance~\cite{BDDCL15}. 

Differently from the Java approach, in~\cite{P11} a minimal core Java  is extended to {\em $\lambda$-expressions} by adding function types, following the style of functional languages. The corresponding type inference algorithm uses sets of constraints and type assumptions, then a substitution operation is required for type variables similar to standard unification. Thus complexity of type inference increases in a substantial way, with respect to Java's one and to the type inference in our calculus. Furthermore, no formal proof of type-safety is provided for this language.

We observe that adding real function types entails that a method must have a different signature according to whether it can accept an object or a function. This sharply contrasts with Java philosophy to continuously fuse language innovations into the old  layer.

Empirical methodologies are used in~\cite{MKTD17} to illustrate when, how and why imperative programmers adopt $\lambda$-expressions.

 \section{Conclusion and Future Work}\label{fw}
We presented the core calculus $\fjil$, which extends a minimal standard model of Java with $\lambda$-expressions and intersection types. Our main intent was to provide a deeper understanding and a formal account for the novel features of Java 8, in order to state and prove related formal properties. A crucial issue has been to design a type system modelling and unifying standard typechecking of object-oriented expressions and type inference for $\lambda$-expressions. Moreover, specific challenges arose to cope with intersection types. 
 As a result, we proved the subject reduction property and progress for  $\fjil$. Since $\fjil$ programs are typed and behave the same as Java programs, our formal result demonstrates that those significant novelties are interwoven in Java 8 in a \emph{type-safe} way. As a by-product of our analysis, we introduced the subtyping rule [$\st\&$L], that, at the best of our knowledge, was never considered before in the present setting, while it is standard in the theory of intersection types (see Part III of~\cite{BDS13}).

Furthermore, we observe that generic functional interfaces are largely used as target types of $\lambda$-expressions, typically the interfaces ${\sf Function <\T,\U>}$ and ${\sf Predicate <\T>} $. The extension of $\fjil$ to generic types poses a significant challenge, since some problems arise from the Java semantics. For instance, a more complicated notion of functional type would be needed to cope with the intersection of generic types, taking into account that method signatures are modified by erasure. Also the definition of the function $\texttt{lub}$ for typing the conditionals would become trickier, as observed in Section~\ref{cond}. Therefore we leave the study of \emph{generic} $\fjil$, based on the core calculus GJ of ~\cite{IPW01}, as future work. 
 
Concluding,  the main takeaway of our formalisation is that we could extend the syntax of $\fjil$ to additional cases 
that allow valid uses of explicit intersection types, while keeping the type checking  straightforward  as in FJ and in $\fjil$.
For instance, it would be interesting in our formal calculus to allow methods to have intersections as formal parameter types, as already proposed in~\cite{BW98} for Java without $\lambda$-expressions.  This would be a sensible feature, since it increases polymorphism 
in method calls, both on objects and on $\lambda$-expressions.
In the latter case, in particular, if a method could have as a formal parameter type an intersection of interfaces, we could pass to the method a $\lambda$-expression, on which we can call default methods belonging to different interfaces.
In future works we aim at investigating extensions of $\fjil$ in this direction, towards a further type-safe evolution of Java.

\section*{Acknowledgement}

 We thank the referees whose suggestions guided us in strongly improving the paper, in particular in amending the definition of functional type and the proof of the substitution lemma.
 
 
 Mariangiola and Betti had the pleasure and privilege of frequenting Furio Honsell over many years. He has always displayed great clarity in honing onto the gist of many subjects. 
  It is however as much for his compelling enthusiasm and humor as for his deep knowledge that we are truly indebted to him. 

\bibliographystyle{alpha}
\bibliography{ref}
\end{document}